\documentclass[11pt]{article}
\usepackage{amsmath,amssymb,amsbsy,amscd,amsthm,dsfont}
\usepackage{amsfonts}
\usepackage{anysize}

\newtheorem{theorem}{Theorem}

\newtheorem{proposition}[theorem]{Proposition}

\newtheorem{lemma}[theorem]{Lemma}

\begin{document}

\title{Groupoids and the tomographic picture of quantum mechanics}
\author{A. Ibort$^a$, V.I. Man'ko$^b$, G. Marmo$^{c,d}$, A. Simoni$^{c,d}$,
C. Stornaiolo$^d$, \and and F. Ventriglia$^{c,d}$ \\
{\footnotesize \textit{$^a$Departamento de Matem\'{a}ticas, Universidad
Carlos III de Madrid, }}\\
{\footnotesize \textit{Av.da de la Universidad 30, 28911 Legan\'{e}s,
Madrid, Spain }}\\
{\footnotesize {(e-mail: \texttt{albertoi@math.uc3m.es})}}\\
{\footnotesize \textit{$^b$P.N.Lebedev Physical Institute, Leninskii
Prospect 53, Moscow 119991, Russia}}\\
{\footnotesize {(e-mail: \texttt{manko@na.infn.it})}}\\
\textsl{{\footnotesize {$^c$Dipartimento di Fisica dell' Universit\`{a}
``Federico II" di Napoli,}}}\\
\textsl{{\footnotesize {Complesso Universitario di Monte S. Angelo, via
Cintia, 80126 Naples, Italy}}}\\
\textsl{{\footnotesize {$^d$ Sezione INFN di Napoli,}}}\\
\textsl{{\footnotesize {Complesso Universitario di Monte S. Angelo, via
Cintia, 80126 Naples, Italy}}}\\
{\footnotesize {(e-mail: \texttt{marmo@na.infn.it, simoni@na.infn.it,
cosmo@na.infn.it, ventriglia@na.infn.it})}}}
\maketitle

\abstract{ The existing relation between the tomographic description of quantum
states and the convolution algebra of certain discrete groupoids
represented on Hilbert spaces will be discussed. The realizations of groupoid algebras  based on qudit,
photon-number (Fock) states and symplectic tomography quantizers and
dequantizers will be constructed. Conditions for identifying the convolution product of groupoid functions and the star--product arising from a quantization--dequantization scheme will be given.   A tomographic approach to construct quasi--distributions out of suitable immersions of groupoids into Hilbert spaces will be formulated and, finally, intertwining kernels for such generalized
symplectic tomograms will be evaluated explicitly.}

\section{Introduction}

Symmetries of quantum and classical systems play an important role in
studying their properties \cite{Beppe,3gbook}. Symmetries are usually
associated with \textquotedblleft symmetry groups\textquotedblright\ and in
the case of continuous variables with Lie groups. Other mathematical
structures very close to groups, like semigroups and groupoids, have been
found to be relevant to describe other properties of a physical system.
Semigroups are very important to describe open quantum systems. As for
groupoids, we may quote for instance from the section \textquotedblleft
Algebra of Physical Quantities\textquotedblright\ in the book on
Noncommutative Geometry \cite{Connes} by Alain Connes:

``The set of frequencies emitted by an atom does not form a group, and it is
false that the sum of two frequencies of the spectrum is again one. What
experiments dictate is the Ritz--Rydberg combination principle which permits
indexing the spectral lines by the set $\Delta $ of all pairs $(i,j)$ of
elements of a set $I$\ of indices. The frequencies $\nu _{(ij)}$\ and \ $\nu
_{(kl)}$ only combine when $j=k$ to yield $\nu _{(il)}=\nu _{(ij)}+\nu
_{(jl)}$ (...). Due to the Ritz--Rydberg combination principle, one is not
dealing with a group of frequencies but rather with a groupoid $\Delta
=\left\{ (i,j);i,j\in I\right\} $ having the composition rule $%
(i,j)(j,k)=(i,k)$. The convolution algebra still has a meaning when one
passes from a group to a groupoid, and the convolution algebra of the
groupoid $\Delta $ is none other that the algebra of matrices since the
convolution product may be written $(ab)_{(i,k)}=\sum_{n}a_{(i,n)}b_{(n,k)}$
which is identical with the product rule of matrices."

Thus a key role for the introduction of groupoids in the description of
quantum systems is played by the combination principle of frequencies which
induces by means of a Fourier expansion, say $a=\sum_{m,n}a_{mn}\exp (%
\mathrm{i}\nu _{mn}t)$ and $b=\sum_{k,j}b_{kj}\exp (\mathrm{i} \nu _{kj}t),$
whose product $ab$ gives the matrix multiplication rule $(ab)_{ik}=%
\sum_{n}a_{in}b_{nk}$ which is just convolution in the groupoid algebra.
Then, even if Heisenberg's matrix mechanics, and its foundations, has only a
historical interest today \cite{VDWaerden}, it is nevertheless worthy to
discuss the role that groupoids play in the foundations of quantum mechanics
because it goes far beyond that of groups. We would like to mention that the
use of groupoids in quantum mechanics has been advocated not only by Connes
but by Accardi \cite{Accardi} too.

Recently, we have considered quantum tomography as providing an alternative
picture of quantum mechanics \cite{Ibort I,Ibort II}. Specifically we have
considered a $C^{\star }-$algebraic approach to quantum tomography where
natural instances of $C^\star-$algebras are provided by group algebras and
the star--product is nothing but the convolution product on the group
itself. As the convolution product makes sense also on groupoids and, as
explained before, the Heisenberg picture originates directly from a groupoid
as a carrier space along with the Ritz--Rydberg combination principle, it is
quite natural to try to compare the tomographic picture with the Heisenberg
one by using the common ground of groupoids.

In quantum mechanics, groups are used practically in all branches of
research. Also, semigroups are important for considering evolution of open
quantum systems undergoing dissipation or decoherence. The groupoid
structure however is not so well known in the physicists community, even
though there are authoritative examples of its use as we have already
stressed. Thus it seems appropriate from a pedagogical point of view, to
compare these three different structures: groups, semigroups and groupoids.
A binary composition is defined for all elements of a group and a semigroup,
while for a groupoid such composition is defined only for particular pairs
of elements and is associative when defined. Moreover, any element of a
group has an inverse element along with the unit element. In a groupoid
there are many different unit elements and any element has its own inverse
which may be composed from the left and from the right, yielding different
units. On the contrary for a semigroup the unit element is unique, but in
general no inverse element exists. Precise definitions and properties of
groupoids will be discussed in section 2 and in Appendix A.

It is the aim of this paper to review the general ingredients for the
existence of a star--product on a measure space and to show that, once
groupoids are realized in simple terms in Hilbert spaces and Fock spaces,
the groupoid algebra convolution coincides with the star--product whenever
\textquotedblleft quantizer-dequantizer\textquotedblright maps are properly
defined. This property is the counterpart for groupoids of what happens for
groups, there the star--product coincides with the group-algebra convolution
when quantizer and dequantizer are given by unitary representations of the
group. Then such construction will be used to construct quantum tomographies
based on groupoid algebras. A more elaborated comparison with tomograms for
quantum field theories will be dealt with in future work.

The paper is organized in the following manner. In sec. 2 we begin with a
short discussion of a simple example of groupoid which will be relevant in
the whole paper. We will define its groupoid algebra and its convolution
product. In sec. 3 we explain the essential ingredients to construct
star--products in the quantization--dequantization scheme. In sec. 4 we
compare the two kinds of products present on the groupoid algebra functions,
i.e. the convolution product and star--products. We will state a proposition
giving the conditions for their coincidence. In sec. 5 we discuss some
examples of Hilbert space realizations of groupoids of physical interest
satisfying the previous conditions, and a counterexample will also be
provided. In sec. 6 contact is made between certain quasi--distributions
arising from realizations of groupoids in the tomographic approach, and
well--known tomographic schemes as spin, photon number and symplectic ones.
Finally, some conclusions and perspectives are drawn in sec. 7. In the
mathematical Appendix more details and example of groupoids are provided,
while the proposition of sec. 4 is slightly generalized, at least for finite
groupoids.

\section{A friendly introduction to Groupoids}

We will start discussing the pair--groupoid of a set, a relevant example for
the physical situations considered afterwards and arguably the simplest
example of a groupoid, leaving a more rigorous and complete treatment of
groupoids to the Appendix. In following this example, it is convenient to
bear in mind the Ritz--Rydberg combination principle.

Given a set $S$ consider the cartesian product
\begin{equation}
\Gamma =S\times S=\left\{ \left( x,y\right) ,x,y\in S\right\} .
\end{equation}%
$\Gamma $ may be given the structure of a groupoid as follows. We define
first the set $G_{0}$ of ``units" of $\Gamma $ as%
\begin{equation}
G_{0}=\left\{ (x,x)\right\} _{x\in S}\subseteq \Gamma
\end{equation}%
and two maps, $r$ and $s,$ from $\Gamma $ onto $G_{0}$
\begin{equation}
r,s:\Gamma \rightarrow G_{0},\quad r(x,y)=(x,x),\quad s(x,y)=(y,y)
\end{equation}%
The composition law (product) $\circ $ of two elements $(x,t),(z,y)$ of $%
\Gamma $ is defined if and only if $t=z,$ i.e. $s(x,t)=r(z,y),$ and reads%
\begin{equation}
(x,t)\circ (t,y)=(x,y).
\end{equation}%
The Ritz--Rydberg frequency indices obey the same composition law. The
product $\circ $ is associative, as one can readily check and $r(x,y)$, $%
s(x,y)$ are the left and right unity for $(x,y)$ respectively.

Finally, every element $(x,y)$ has a (left and right) inverse given by%
\begin{equation}
(x,y)^{-1}=(y,x)
\end{equation}%
and one has%
\begin{equation}
(x,y)\circ (y,x)=(x,x)\quad ,\quad (y,x)\circ (x,y)=(y,y).
\end{equation}%
Moreover%
\begin{equation}
(y,x)\circ \left( (x,y)\circ \left( y,z\right) \right) =\left( y,z\right)
\quad ,\quad \left( \left( z,x\right) \circ \left( x,y\right) \right) \circ
(y,x)=\left( z,x\right) .
\end{equation}

Consider the map%
\begin{equation}
\left( r,s\right) :\Gamma \rightarrow G_{0}\times G_{0},\quad \left(
x,y\right) \mapsto \left( \left( x,x\right) ,\left( y,y\right) \right) .
\end{equation}%
This map is both one--to--one and onto, so the groupoid $\Gamma $\ is said
to be principal and transitive, respectively.

A function on $\Gamma $ is a (real or complex) function $f(x,y)$. When $S$
is a finite or a countable set, say $\{x_{i}\}$ with $i,k=1,..,N\leq \infty$%
, a function $f(x_{i},x_{k})$ may be thought of as a a finite or an infinite
matrix. By thinking of $f(x_{i},x_{k})$ in these terms makes available the row--by--column
product. It is now easy to see that by using an irreducible groupoid realization,
the matrices we associate with any function constitute a realization of the groupoid algebra.
Let us expand this statement.

We use an irreducible groupoid realization:
the element $(x_{i},x_{k})$ is associated with the $N\times N-$matrix $%
E_{ik},$ whose entries are zero, but the entry $ik$ which is one:%
\begin{equation}
\left( E_{ik}\right) _{jl}=\delta _{ij}\delta _{kl}.
\end{equation}%
Moreover,%
\begin{equation}
E_{ik}E_{jl}=\delta _{kj}E_{il}  \label{e prod}
\end{equation}%
so that if $(x_{i},x_{k})$ and $(x_{j},x_{l})$ are not composable the
product of their representative matrices is zero. This generalizes a group
realization to the groupoid case. The transpose $E_{ik}^{T}$ of $E_{ik}$ is $%
E_{ki},$ associated with $(x_{k},x_{i}),$ the inverse of $(x_{i},x_{k})$.
The matrices $E^{\prime }$s are an orthonormal basis of the corresponding
Hilbert--Schmidt space, as
\begin{equation}
\mathrm{Tr}\left[ E_{ik}E_{lj}^{T}\right] =\delta _{il}\delta _{kj}.
\label{orth}
\end{equation}%
In the present case, they are just the \textquotedblleft units" defined by
Weyl \cite{weyl} long ago, while introducing the concept of group algebra
and its realizations. They provide an irreducible realization of the
groupoid algebra functions in terms of $N\times N-$matrices:%
\begin{equation}
f\mapsto A_{f},\quad A_{f}:=\sum_{i,k=1}^{N}f(x_{i},x_{k})E_{ik}
\label{real group alg}
\end{equation}

We also notice, in passing, that these aspects contain also the basic
ingredients of the Schwinger's measurements algebra \cite{Schwinger}. The
groupoid algebra can be defined, analogously to Weyl's definition of a group
algebra, as the algebra of $\ N^{2}-\dim $ vectors, associated with the
functions on a groupoid. So, the vector $F,$ associated with the function $f$%
, has component $f(x_{i},x_{k})$ in the \textquotedblleft direction" $%
(x_{i},x_{k}).$ Such a vector can be expressed as a formal linear
combination:%
\begin{equation}
F=\sum_{i,k=1}^{N}\left[ f(x_{i},x_{k})\right] (x_{i},x_{k}),
\end{equation}%
so that the groupoid realization in terms of units $(x_{i},x_{k})\mapsto
E_{ik\text{ }}$yields at once the irreducible realization of the groupoid
algebra of eq. \eqref{real group alg}. The groupoid algebra product can be
defined as
\begin{equation}
F_{1}F_{2}=\sum\limits_{\overset{i,j,k,l=1}{\mathrm{allowed~}%
(x_{i},x_{j})\circ (x_{k},x_{l})}}^{N}\left[
f_{1}(x_{i},x_{j})f_{2}(x_{k},x_{l})\right] (x_{i},x_{j})\circ (x_{k},x_{l})
\end{equation}%
As a consequence, the related algebra of groupoid functions is equipped with
a convolution product, which reads
\begin{equation}
\left( f_{1}\ast f_{2}\right) (x_{n},x_{m})=\sum\limits_{\overset{i,j,k,l=1}{%
(x_{i},x_{j})\circ (x_{k},x_{l})=(x_{n},x_{m})}%
}^{N}f_{1}(x_{i},x_{j})f_{2}(x_{k},x_{l}).
\end{equation}%
Taking into account that $(x_{i},x_{j})\circ (x_{k},x_{l})$ is defined only
if $j=k$ and\ $(x_{i},x_{k})\circ (x_{k},x_{l})=(x_{n},x_{m})$ gives $i,l=m,$
these constraints may be implemented in terms of Kronecker delta symbols,
and the convolution formula becomes:
\begin{eqnarray}
\left( f_{1}\ast f_{2}\right) (x_{n},x_{m})
&=&\sum\limits_{i,j,k,l=1}^{N}f_{1}(x_{i},x_{j})f_{2}(x_{k},x_{l})\delta
_{in}\delta _{jk}\delta _{lm}  \label{conv} \\
&=&\sum\limits_{k=1}^{N}f_{1}(x_{n},x_{k})f_{2}(x_{k},x_{m}).  \notag
\end{eqnarray}

In terms of the representative matrices $A^{\prime }$s, the above
convolution formula is realized as the usual\ row--by--column matrix
product. In view of eq. \eqref{e prod}, we get%
\begin{eqnarray}
A_{f_{1}}A_{f_{2}}
&=&\sum_{i,k,j,l=1}^{N}f_{1}(x_{i},x_{k})f_{2}(x_{j},x_{l})E_{ik}E_{jl}=%
\sum_{i,k,l=1}^{N}f_{1}(x_{i},x_{k})f_{2}(x_{k},x_{l})E_{il} \\
&=&\sum_{i,l=1}^{N}\left( f_{1}\ast f_{2}\right)
(x_{i},x_{l})E_{il}=A_{f_{1}\ast f_{2}}.  \notag
\end{eqnarray}

In conclusion, as for groups, the groupoid structure yields the possibility
of introducing an algebra structure in the set $\mathcal{F}\left( \Gamma
\right) $ of the groupoid functions. Besides, for this kind of principal
groupoids, functions are associated with numerical matrices, the groupoid
algebra can be realized as the algebra of $N\times N-$matrices and
convolution is the usual matrix product.

Finally, when the groupoid is not countable, the convolution formula %
\eqref{conv} entails an integration over the groupoid with respect to a
suitable measure. For instance, in the continuous case $S=\mathbb{R}$, it is
natural to generalize matrix multiplication to a \textquotedblleft
continuous matrix multiplication \textquotedblright
\begin{equation}
\left( f_{1}\ast f_{2}\right) (x,y)=\int_{\mathbb{R}}f_{1}(x,t)f_{2}(t,y)%
\mathrm{d}t  \label{contmat}
\end{equation}%
where $\mathrm{d}t$\ is the translation--invariant Lebesgue measure and $%
f_{1},f_{2}$ are continuous functions of proper support \cite{weinstein}, we
postpone further details to the Appendix C.

\section{A general framework for a $\star -$product}

Given a vector space $V$ (for the moment finite-dimensional) and its dual $%
V^{\ast },$ we can consider immersions of a measure space $(X,\mathrm{d}x)$
into $V$ and $V^{\ast },$ say $x\mapsto v\left( x\right) $ and $x\mapsto
\alpha \left( x\right) ,$ satisfying the relation (which is written using
different notation):%
\begin{equation}
\mathbb{I}_{V} = \int_{X} v\left( x\right) \mathrm{d}x \, \alpha \left(
x\right) \equiv \int_{X}v_{x}\, \mathrm{d}x \, \alpha _{x}\equiv
\int_{X}v_{x}\otimes \alpha _{x} \, \mathrm{d}x,  \label{decid}
\end{equation}%
i.e., we can construct a resolution of the identity $\mathbb{I}_{V}\in
\mathrm{End}\left( V\right) $ with the property%
\begin{equation*}
\alpha _{y}=\int_{X}\alpha _{y}\left( v_{x}\right) \alpha _{x}\mathrm{d}x.
\end{equation*}

Once the previous immersions are given, it is possible to associate with any
vector $w\in V$ a function $f_{w}$ on $X$ as%
\begin{equation}
f_{w}\left( x\right) :=\alpha _{x}\left( w\right)  \label{deq}
\end{equation}%
With any function $f$ in the image of this map we may associate a vector $%
w_{f}$ in $V$ by setting%
\begin{equation}
w_{f}:=\int_{X}f\left( x\right) v_{x}\mathrm{d}x.  \label{qu}
\end{equation}

Assuming, for simplicity, that any function on $X$ is in the image of the
map $w\mapsto f_{w}$, it is possible to endow the space of functions on $X$
with all additional structures we may have on $V$. For instance, if $V$ is a
Lie algebra with Lie bracket $\left[ \cdot ,\cdot \right] $ we may define%
\begin{equation}
\left[ f_{v_{1}},f_{v_{2}}\right] :=f_{\left[ v_{1},v_{2}\right] }.
\end{equation}%
Similarly, if $V$ is an associative algebra with product $B \colon V\times
V\rightarrow V$, $v\cdot u = B \left( v, u\right) \in V $ and $v_i$ is a
linear basis, then:

\begin{equation}
v_{j}\cdot v_{k} = \sum_{l}c_{jk}^{l}v_{l},
\end{equation}%
and we may define%
\begin{equation}
f_{v_{j}}\star f_{v_{k}}:=f_{v_{j}\cdot v_{k}}=f_{B\left( v_{j},v_{k}\right)
}.
\end{equation}

This abstract general framework has been extensively analyzed in previous
work \cite{quant-dequant} under the evocative name of quantizer--dequantizer
for the algebra of operators on some Hilbert space $\mathcal{H}$, i.e., to
be more precise, we may consider the Hilbert space $V$ of Hilbert-Schmidt
operators on $\mathcal{H}$ and, after identifying $V$ and $V^{\ast }$ in a
natural way, the immersions $v(x)$ and $\alpha (x)$ above will be written as
$\hat{D}(x)$ and $\hat{U}(x)$ respectively. We recall that in such a case,
the dequantizer $\hat{U}\left( x\right) $ and quantizer $\hat{D}\left(
x\right) $ give rise to formulae%
\begin{equation}
f_{A}\left( x\right) =\mathrm{Tr}\left[ \hat{A}\hat{U}^{\dagger }\left(
x\right) \right]  \label{dequantU}
\end{equation}%
instead of eq. \eqref{deq} and%
\begin{equation}
\hat{A}=\int_{X}f_{A}\left( x\right) \hat{D}\left( x\right) \mathrm{d}x
\end{equation}%
instead of eq. \eqref{qu}, while eq. \eqref{decid} becomes%
\begin{equation}
\mathrm{Tr}\left[ \hat{D}\left( x\right) \hat{U}^{\dagger }\left( x^{\prime
}\right) \right] =\delta _{x}\left( x^{\prime }\right) .
\end{equation}%
Symbol functions $f_{A}$ may be composed by a star--product kernel:

\begin{equation}
\left( f_{A}\star f_{B}\right) \left( x\right) =\int f_{A}\left(
x_{1}\right) f_{B}\left( x_{2}\right) K\left( x_{1},x_{2},x\right) \mathrm{d}%
x_{1}\mathrm{d}x_{2}=f_{AB}\left( x\right) ,
\end{equation}%
where%
\begin{equation}
K\left( x_{1},x_{2},x\right) :=\mathrm{Tr}\left[ \hat{D}\left( x_{1}\right)
\hat{D}\left( x_{2}\right) \hat{U}^{\dagger }\left( x\right) \right]
\label{ker}
\end{equation}

Previous abstract general framework has also been exploited in the
tomographic picture of quantum mechanics, where with any vector of the
Hilbert space we associate a probability distribution and the space $X$ is
usually associated with the point spectrum of a maximal set of commuting
operators. An additional family of parameters is often provided by a Lie
group $G$ acting by similarity transformations on the selected maximal
Abelian set of operators, in such a manner that the union of the transformed
operators for various elements of $G$ constitutes a tomographic set.

When the space $X$ is an Abelian vector group and the immersion is
associated with a Weyl system, the outlined scheme includes the Wigner--Weyl
picture along with the Moyal formalism \cite{moyal49,Gr}. Other instances
are provided by (generalized) coherent states \cite{coherent}.

This picture, when the vector space $V$ is an associative algebra, provides
an interesting way to construct associative products on the space of
functions $\mathcal{F}\left( X \right)$ by means of a kernel function $K
(x_{1},x_{2},x_{3}),$ representing the ``density" of the associative product
$B$, which satisfies the quadratic condition originated by associativity.

When the measure space $(X,\mathrm{d}x)$ carries additional structures which
are ``compatible" with the immersion in the algebra space of operators,
there are interesting and useful properties inherited by the constructed
star--product. For instance, we may look for the group algebra of a Lie
group $G$ and for the convolution product on $\mathcal{F} \left( G \right).$
In this sense, it would be possible to generalize the picture emerging in
the Weyl--Wigner formalism to generic groups.

As a particular generalization, because the convolution product makes sense
also for functions defined on groupoids, we could like to analyze to what
extent the Weyl--Wigner picture generalizes also to groupoids. In summary,
the problem we are analyzing may be formulated in the following manner: To
study the properties of the star--product when the carrier space carries the
structure of a groupoid and the immersion into the space of operators may be
considered to be compatible with the product structure.

To this aim, we first need to study some properties of groupoids coming back
to the example of sec. 2.

\section{Groupoids: $\ast -$convolution and $\star -$product}

To get more insight, we restrict first to the case of the principal
transitive pair--groupoid $\Gamma =S\times S$\ , with $S$ a countable set of
order $N\leq \infty ,$ discussed in Sec. 2. We begin by observing that the
groupoid algebra realization eq. \eqref{real group alg} looks like a
quantization formula, where an operator $A_{f}$ is associated with a
function $f$ by means of a quantizer $E$. This is not a mere analogy: in
fact, in view of eq. \eqref{orth}, the dequantizer associated with $E_{ik}$
is just $E_{ik}$ (compare with eq. \eqref{dequantU}). In other words, we
have a self--dual pair of quantizer--dequantizer. So, we get the
dequantization formula%
\begin{equation}
\mathrm{Tr}\left[ A_{f}E_{lj}^{T}\right] =\mathrm{Tr}\left[
\sum_{i,k=1}^{N}f(x_{i},x_{k})E_{ik}E_{lj}^{T}\right] =%
\sum_{i,k=1}^{N}f(x_{i},x_{k})\delta _{il}\delta _{kj}=f(x_{l},x_{j}).
\end{equation}%
Thanks to the quantization--dequantization scheme based on Weyl units,
functions on the groupoid can be multiplied by means of a $\star -$product
defined by the kernel associated with the quantizer--dequantizer pair. Thus
we can state the following:

\begin{proposition}
\label{propo} The $\star -$product corresponding to the kernel
\begin{equation*}
K\left( x_{i},x_{k};x_{j},x_{l};x_{m},x_{n}\right) =\mathrm{Tr}\left[
E_{ik}E_{jl}E_{mn}^{T}\right]
\end{equation*}%
coincides with the $\ast-$convolution product of the pair groupoid algebra.
\end{proposition}

\begin{proof}  In fact, the kernel for the $\star-$product can be readily evaluated by using eqs. \eqref{e prod}-\eqref{orth}
\begin{equation}
K\left( x_{i},x_{k};x_{j},x_{l};x_{m},x_{n}\right) =\mathrm{Tr}\left[
E_{ik}E_{jl}E_{mn}^{T}\right] =\delta _{kj}\mathrm{Tr}\left[ E_{il}E_{mn}^{T}%
\right] =\delta _{kj}\delta _{ln}\delta _{im},
\end{equation}%
so that the $\star -$product reads%
\begin{eqnarray}
\left( f_{1}\star f_{2}\right) (x_{m},x_{n})
&=&\sum\limits_{i,j,k,l=1}^{N}f_{1}(x_{i},x_{k})f_{2}(x_{j},x_{l})\delta
_{kj}\delta _{ln}\delta _{im} \\
&=&\sum\limits_{k=1}^{N}f_{1}(x_{m},x_{k})f_{2}(x_{k},x_{n})=\left(
f_{1}\ast f_{2}\right) (x_{m},x_{n})  \notag
\end{eqnarray}%
and coincides with the $\ast -$convolution product.
\end{proof}

The simple proposition is one of the observations of this paper. In the next
sections it will be checked by discussing some instances of physical
interest and will be generalized to an example of Lie groupoids. In Appendix
B, see Prop. \ref{gen_conv}, this result will be generalized further to an
arbitrary transitive, possibly not principal, groupoid of finite order.

\section{Immersions of groupoids into Hilbert spaces}

As can be easily inferred from sec. 2, formulae \eqref{e prod},\eqref{orth},
it is clear that any orthononormal basis of $\mathcal{H},$\ say associated
with a maximal set of commuting operators, can be used to explicitly
construct "realizations" of groupoids. Let see some specific example.

We begin with a finite dimensional case. Consider the finite set $%
S_{j}=\left\{ m,-j\leq m\leq j\right\} $, $j$ an integer or semi-integer
number, with the counting measure $\mu _{j}\left( \{m\}\right) =1$ for any $%
m $. The principal transitive pair--groupoid $\Gamma _{j}=S_{j}\times S_{j}$%
, with the composition law $\left( m,m^{\prime }\right) \circ \left(
m^{\prime },m^{\prime \prime }\right) =\left( m,m^{\prime \prime }\right) $
and the product measure $\mu _{j}\times \mu _{j}$, is a measure space which
we immerse into the Hilbert space of operators $B\left( \mathcal{H}%
_{j}\right) $ of the $\left( 2j+1\right) -$dimensional Hilbert space \ $%
\mathcal{H}_{^{j}}$ by means of:%
\begin{equation}
\left( m,m^{\prime }\right) \in \Gamma _{j}\longmapsto \left\vert
jm\right\rangle \left\langle jm^{\prime }\right\vert \in B\left( \mathcal{H}%
_{j}\right)
\end{equation}%
where $\left\vert jm\right\rangle ,\left\vert jm^{\prime }\right\rangle $
are the canonical orthonormal eigenvectors of the angular momentum operators
$\hat{J}^{2},$ $\hat{J}_{z}$; $\hat{J}^{2}\left\vert jm\right\rangle
=j(j+1)\left\vert jm\right\rangle $, $\hat{J}_{z}\left\vert jm\right\rangle
=m\left\vert jm\right\rangle $. Dropping the label $j,$ we may write the
completeness relation%
\begin{equation}
\sum_{m=-j}^{j}\left\vert m\right\rangle \left\langle m\right\vert =\mathbb{I%
}_{j},
\end{equation}%
which eventually gives
\begin{equation}
\hat{A}=\sum_{m,m^{\prime }=-j}^{j}\tau _{A}^{j}\left( m,m^{\prime }\right)
\left\vert m\right\rangle \left\langle m^{\prime }\right\vert ,
\label{spin quant}
\end{equation}%
where
\begin{equation}
\tau _{A}^{j}\left( m,m^{\prime }\right) =\mathrm{Tr}\left[ \hat{A}\left(
\left\vert m\right\rangle \left\langle m^{\prime }\right\vert \right)
^{\dagger }\right] =\langle m|\hat{A}|m^{\prime }\rangle .
\end{equation}%
The above equations provide a self--dual pair of quantizer--dequantizer%
\begin{equation}
\hat{D}_{j}\left( m,m^{\prime }\right) =\hat{U}_{j}\left( m,m^{\prime
}\right) =\left\vert m\right\rangle \left\langle m^{\prime }\right\vert .
\end{equation}%
As the Weyl basis operators $\left\{ \left\vert m\right\rangle \left\langle
m^{\prime }\right\vert \right\} $ are associated with the Weyl units $%
\left\{ E_{mm^{\prime }}\right\} ,$ because of Proposition \ref{propo} the
star--product for the symbols $\tau ^{\prime }$s coincides with the
convolution product of the groupoid algebra of $\Gamma _{j}$.

By means of a direct sum over the label $j$ we may construct the infinite
dimensional case starting with the countable groupoid:
\begin{equation*}
\Gamma =\bigsqcup _{j}\Gamma _{j},
\end{equation*}
with composition law:
\begin{equation*}
\left( jm,jm^{\prime }\right) \circ \left( j^{\prime }m^{\prime },j^{\prime
}m^{\prime \prime }\right) =\left( jm,jm^{\prime \prime }\right), \quad
\mathrm{iff} \quad j^{\prime }=j,
\end{equation*}
and immersion into the space of Hilbert--Schmidt operators $\mathcal{I}%
_{2}\subset B\left( \mathcal{H}\right)$, $\mathcal{H}=\bigoplus _{j}\mathcal{%
H}_{j},$ given by
\begin{equation}
\left( jm,jm^{\prime }\right) \in \Gamma \longmapsto \left\vert
jm\right\rangle \left\langle jm^{\prime }\right\vert \in B\left( \mathcal{H}
\right) .
\end{equation}

The groupoid $\Gamma ,$\ as disjoint union of groupoids, is principal
nontransitive and the immersion decomposes trivially into the finite
dimensional cases previously analyzed. A less trivial infinite dimensional
case is the following.

Let us consider the discrete space $\mathbb{N}_{0}=\{0,1,2,...,\}$ with
counting measure $\mu \left( \{n\}\right) =1$ for any natural number $n$.
Consider again the pair--groupoid $\Gamma =\mathbb{N}_{0}\times \mathbb{N}%
_{0}$\ with the product measure $\mu \times \mu ,$ it is a principal
transitive groupoid with composition rule $\left( n,m\right) \circ \left(
m,p\right) =\left( n,p\right) $. Given a complex separable Hilbert space $%
\mathcal{H}$, we can immerse the measure space $(\Gamma ,\mu \times \mu )$
into the Hilbert space $\mathcal{I}_{2}$ of Hilbert-Schmidt operators on $%
\mathcal{H}$ as:%
\begin{equation}
\left( n,m\right) \in \Gamma =\mathbb{N}_{0}\times \mathbb{N}_{0}\longmapsto
\left\vert n\right\rangle \left\langle m\right\vert \in \mathcal{I}_{2},
\end{equation}%
where $\left\vert n\right\rangle ,\left\vert m\right\rangle $ are Fock boson
states, i.e. eigenvectors of the number operator $\hat{N}$ associated with a
harmonic oscillator:%
\begin{equation}
\hat{N}\left\vert n\right\rangle \left\vert n\right\rangle ,\quad
\left\langle n|m\right\rangle =\delta _{nm}.
\end{equation}%
The completeness relation%
\begin{equation}
\sum_{n}\left\vert n\right\rangle \left\langle n\right\vert =\mathbb{I}
\end{equation}%
allows for writing%
\begin{eqnarray}
\hat{A} &=&\sum_{n,m}F_{A}\left( n,m\right) \left\vert n\right\rangle
\left\langle m\right\vert ,  \label{represA} \\
F_{A}\left( n,m\right) &=&\mathrm{Tr}\left[ \hat{A}\left( \left\vert
n\right\rangle \left\langle m\right\vert \right) ^{\dagger }\right] =\langle
n|\hat{A}|m\rangle ,
\end{eqnarray}%
so that the Weyl basis $\left\vert n\right\rangle \left\langle m\right\vert $
provides again a self--dual pair of quantizer--dequantizer%
\begin{equation}
\hat{D}\left( n,m\right) =\hat{U}\left( n,m\right) =\left\vert
n\right\rangle \left\langle m\right\vert .
\end{equation}%
The corresponding kernel of the star--product for the symbols $F,$ given by
eq. \eqref{ker} is expressed in terms of Kronecker deltas as%
\begin{eqnarray}
K\left( n_{1},m_{1};n_{2},m_{2};n,m\right) &=&\mathrm{Tr}\left[ \hat{D}%
\left( n_{1},m_{1}\right) \hat{D}\left( n_{2},m_{2}\right) \hat{U}^{\dagger
}\left( n,m\right) \right] \\
&=&\delta _{nn_{1}}\delta _{m_{1}n_{2}}\delta _{m_{2}m},
\end{eqnarray}%
and the star--product reads%
\begin{eqnarray}
\left( F_{A}\star F_{B}\right) \left( n,m\right) &=&\sum_{\overset{%
n_{1},m_{1}}{_{n_{2}},_{m_{2}}}}F_{A}\left( n_{1},m_{1}\right) F_{B}\left(
n_{2},m_{2}\right) K\left( n_{1},m_{1};n_{2},m_{2};n,m\right)  \notag \\
&=&\sum_{m_{1}}F_{A}\left( n,m_{1}\right) F_{B}\left( m_{1},m\right) .
\end{eqnarray}%
which is just the convolution product of the groupoid algebra of $\Gamma =%
\mathbb{N}_{0}\times \mathbb{N}_{0}$ given by the usual product of matrices
in agreement with the result of Proposition \ref{propo}.

As a second example, we will consider given an operator $\hat{A}$, a new
symbol $f_A(x,y)$ defined by:
\begin{eqnarray}
f_{A}\left( x,y\right) &=&\sum_{n,m}F_{A}\left( n,m\right) \varphi
_{n}\left( x\right) \varphi _{m}^{\ast }\left( y\right) \\
&=&\mathrm{Tr}\left[ \sum_{n,m}F_{A}\left( n,m\right) \left( \left\vert
n\right\rangle \left\langle m\right\vert \right) \left( \left\vert
x\right\rangle \left\langle y\right\vert \right) ^{\dagger }\right] =\mathrm{%
Tr}\left[ \hat{A}\left( \left\vert x\right\rangle \left\langle y\right\vert
\right) ^{\dagger }\right] .  \notag
\end{eqnarray}%
Here $\varphi _{n}\left( x\right) :=\left\langle x|n\right\rangle $ is the $%
n- $th Hermite function, i.e. the $n-$th excited state wave function of the
harmonic oscillator in the position representation and $\left\vert
x\right\rangle ,\left\vert y\right\rangle $ are (improper) eigenvectors of
the position operator $\hat{q}:$%
\begin{equation}
\hat{q}\left\vert x\right\rangle =x\left\vert x\right\rangle ,\quad
\left\langle x|y\right\rangle =\delta \left( x-y\right) .
\end{equation}%
\textit{Vice versa}, the old symbol $F_A (n,m)$ can be recovered from the
new one as:%
\begin{equation}
F_{A}\left( n,m\right) =\int \mathrm{d}x\mathrm{d}y \, f_{A}\left(
x,y\right) \varphi _{n}^{\ast }\left( x\right) \varphi _{m}\left( y\right) .
\end{equation}

The symbol $f_{A}\left( x,y\right) $ may also be obtained considering the
pair groupoid $\Gamma =\mathbb{R}\times \mathbb{R}$\ with the product
measure $\mathrm{d}x\mathrm{d}y$. We may think of $\Gamma = \mathbb{R}\times
\mathbb{R}$ as a principal transitive groupoid $\Gamma $ with composition
rule $\left( x,y\right) \circ \left( y,z\right) =\left( x,z\right) $ and
choose the immersion:%
\begin{equation}
\left( x,y\right) \mathbb{\longmapsto }\left\vert x\right\rangle
\left\langle y\right\vert ,
\end{equation}%
in the space of rank--one operators in the Gelfand triple $\mathcal{S}(%
\mathbb{R}) \subset L^2(\mathbb{R}) \subset \mathcal{S}^{\prime}(\mathbb{R})$%
. Thanks to the completeness relation%
\begin{equation}
\int \mathrm{d}x\left\vert x\right\rangle \left\langle x\right\vert =\mathbb{%
I}
\end{equation}%
we may write, for a given operator $\hat{A}$, formulae similar to %
\eqref{represA},%
\begin{eqnarray}
\hat{A} &=&\int \mathrm{d}x\mathrm{d}y \, \, f_{A}\left( x,y\right)
\left\vert x\right\rangle \left\langle y\right\vert ,  \label{cont rec} \\
f_{A}\left( x,y\right) &=&\mathrm{Tr}\left[ \hat{A}\left( \left\vert
x\right\rangle \left\langle y\right\vert \right) ^{\dagger }\right] =
\langle x|\hat{A}|y \rangle .
\end{eqnarray}%
In other words, the operator $\left\vert x\right\rangle \left\langle
y\right\vert $ provide a self--dual pair of quantizer--dequantizer
\begin{equation}
\hat{D}\left( x,y\right) =\hat{U}\left( x,y\right) =\left\vert
x\right\rangle \left\langle y\right\vert ,
\end{equation}%
satisfying%
\begin{equation}
\mathrm{Tr}\left[ \hat{D}\left( x,y\right) \hat{U}^{\dagger }\left(
x^{\prime },y^{\prime }\right) \right] =\delta \left( x-x^{\prime }\right)
\delta \left( y-y^{\prime }\right) .
\end{equation}%
The kernel of the corresponding star--product for the symbols, eq. %
\eqref{ker}, is expressed in terms of Dirac delta functions:%
\begin{eqnarray}
K\left( x_{1},y_{1};x_{2},y_{2};x,y\right) &=&\mathrm{Tr}\left[ \hat{D}%
\left( x_{1},y_{1}\right) \hat{D}\left( x_{2},y_{2}\right) \hat{U}^{\dagger
}\left( x,y\right) \right] \\
&=&\delta \left( x-x_{1}\right) \delta \left( y_{1}-x_{2}\right) \delta
\left( y_{2}-y\right) .  \notag
\end{eqnarray}%
Thus the star--product reads:%
\begin{eqnarray}
\left( f_{A}\star f_{B}\right) \left( x,y\right) &=&\int \mathrm{d}x_{1}%
\mathrm{d}y_{1}\mathrm{d}x_{2}\mathrm{d}y_{2} \, \, f_{A}\left(
x_{1},y_{1}\right) f_{B}\left( x_{2},y_{2}\right) K\left(
x_{1},y_{1};x_{2},y_{2};x,y\right) \\
&=& \int \mathrm{d}y_{1}\, \, f_{A}\left( x,y_{1}\right) f_{B}\left(
y_{1},y\right) .  \notag
\end{eqnarray}%
In conclusion, we have recovered the usual product of matrices with
continuous labels given by eq. \eqref{contmat}, which is the convolution
product of the groupoid algebra of $\Gamma =\mathbb{R}\times \mathbb{R}.$

Finally, for the three cases considered above, we remark that the
generalization to the multi--mode case $\Gamma =X\times X\times ...\times X$
is straightforward. In fact, as a fourth example, consider a two--mode
harmonic oscillator. We have to immerse $X\times X$ onto a Weyl basis
generated by the oscillator eigenstates $\left\vert n_{1},n_{2}\right\rangle
=\left\vert n_{1}\right\rangle \otimes \left\vert n_{2}\right\rangle :$%
\begin{equation}
\left( \left( n_{1},n_{2}\right) ,\left( m_{1},m_{2}\right) \right) \mapsto
\left\vert n_{1},n_{2}\right\rangle \left\langle m_{1},m_{2}\right\vert
\end{equation}%
This immersion realizes the groupoid combination rule
\begin{equation*}
\left( \left( n_{1},n_{2}\right) ,\left( m_{1}^{\prime },m_{2}^{\prime
}\right) \right) \circ \left( \left( m_{1}^{\prime },m_{2}^{\prime }\right)
,\left( m_{1},m_{2}\right) \right) =\left( \left( n_{1},n_{2}\right) ,\left(
m_{1},m_{2}\right) \right)
\end{equation*}%
yielding $0$ for the forbidden groupoid products. The self--dual
quantizer--dequantizer pair reads:%
\begin{equation}
\hat{D}\left( n_{1},n_{2};m_{1},m_{2}\right) =\hat{U}\left(
n_{1},n_{2};m_{1},m_{2}\right) =\left\vert n_{1},n_{2}\right\rangle
\left\langle m_{1},m_{2}\right\vert ,
\end{equation}%
and we obtain at once%
\begin{equation}
\hat{A}=\sum_{n_{1},n_{2},m_{1},m_{2}}F_{A}\left(
n_{1},n_{2};m_{1},m_{2}\right) \left\vert n_{1},n_{2}\right\rangle
\left\langle m_{1},m_{2}\right\vert ,
\end{equation}%
where%
\begin{equation}
F_{A}\left( n_{1},n_{2};m_{1},m_{2}\right) =\mathrm{Tr}\left[ \hat{A}\, \hat{%
U}^{\dagger }\left( n_{1},n_{2};m_{1},m_{2}\right) \right] = \langle
n_{1},n_{2}\mid \hat{A}\, \mid m_{1},m_{2} \rangle .
\end{equation}%
The kernel of the star--product of the symbol functions can be readily
evaluated in the usual way, and we get%
\begin{equation}
\left( F_{A}\star F_{B}\right) \left( n_{1},n_{2};m_{1},m_{2}\right)
=\sum_{m_{1}^{\prime },m_{2}^{\prime }}F_{A}\left( n_{1},n_{2};m_{1}^{\prime
},m_{2}^{\prime }\right) F_{B}\left( m_{1}^{\prime },m_{2}^{\prime
};m_{1},m_{2}\right) ,
\end{equation}%
which is again a matrix product, as expected.

\medskip

We will end up this section by exhibiting a situation where Prop. \ref{propo}
will not apply. In contrast with the previous cases, let us consider now the
product groupoid $\Gamma =\mathbb{R}^{2}\times \mathbb{R}^{2}\simeq \mathbb{C%
}\times \mathbb{C},\left( \left( x,y\right) ,\left( u,v\right) \right)
\simeq \left( z,w\right) ,$ with $z=x+iy,w=u+iv$. The immersion of $\Gamma $%
\ given by $\left( z,w\right) \longmapsto \left\vert z\right\rangle
\left\langle w\right\vert$, where $\left\vert z\right\rangle ,\left\vert
w\right\rangle $ are coherent states of a harmonic oscillator, i.e., $\hat{a}
| z \rangle = z | z \rangle$ with $\hat{a}, \hat{a}^\dagger$ creation and
annihilation operators, yields a quantizer--dequantizer self--dual pair $%
\hat{D}\left( z,w\right) =\hat{U}\left( z,w\right) =\left\vert
z\right\rangle \left\langle w\right\vert .$ The corresponding star--product
does not coincide with the groupoid algebra convolution, due to the
non--orthogonality of the coherent states. In fact, from the completeness
relation%
\begin{equation}
\int \frac{\mathrm{d}^{2}z}{\pi }\left\vert z\right\rangle \left\langle
z\right\vert =\mathbb{I}
\end{equation}%
we get equations similar to \eqref{represA} again:%
\begin{eqnarray}
\hat{A} &=&\int \frac{\mathrm{d}^{2}z}{\pi }\frac{\mathrm{d}^{2}w}{\pi } \,
\, \mathrm{f}_{A}\left( z,w\right) \left\vert z\right\rangle \left\langle
w\right\vert , \\
\mathrm{f}_{A}\left( z,w\right) &=&\mathrm{Tr}\left[ \hat{A}\left(
\left\vert z\right\rangle \left\langle w\right\vert \right) ^{\dagger }%
\right] = \langle z|\hat{A}|w \rangle .  \notag
\end{eqnarray}%
with star--product kernel%
\begin{eqnarray}
K\left( z_{1},w_{1};z_{2},w_{2};z,w\right) &=&\mathrm{Tr}\left[ D\left(
z_{1},w_{1}\right) D\left( z_{2},w_{2}\right) U^{\dagger }\left( z,w\right) %
\right] \\
&=&\left\langle z|z_{1}\right\rangle \left\langle w_{1}|z_{2}\right\rangle
\left\langle w_{2}|w\right\rangle ,  \notag
\end{eqnarray}%
where%
\begin{equation}
\left\langle w|z\right\rangle =\exp \left[ -\frac{1}{2}\left( \left\vert
w\right\vert ^{2}+\left\vert z\right\vert ^{2}\right) +w^{\ast }z\right]
\neq \delta ^{\left( 2\right) }\left( z-w\right) \equiv \delta \left(
x-u\right) \delta \left( y-v\right) .
\end{equation}%
On the contrary, the groupoid algebra convolution has a kernel expressed in
terms of Dirac delta functions.

\section{Groupoid algebras and tomograms}

The symbols defined in the previous section, $f_{A}\left( x,y\right) ,$ $%
F_{A}\left( n,m\right) $ and $\tau _{A}^{j}\left( m.m^{\prime }\right) $ may
be considered quasi--distributions. When the operator $\hat{A}$ is
Hermitian, positive with trace one, it can be associated with a state,
however, in general its symbol cannot be interpreted as a probability
distribution. Their diagonal parts however, say $f_{A}\left( x,x\right) $, $%
F_{A}\left( n,n\right) $ and $\tau _{A}^{j}\left( m,m\right) $ are positive
but are not sufficient to reconstruct $\hat{A}$. In fact, the off--diagonal
terms are components on base elements orthogonal to the diagonal ones.
However, in the tomographic picture of quantum mechanics, it is possible to
\textquotedblleft rotate" the diagonal base elements in such a way that they
form a family of bases, possibly overcomplete, which allows for a state
reconstruction from the diagonal part of their symbols. These families very
often are generated by acting with a representation of a group, whose
elements are labeled by a set of parameters which appear as independent
variables in the tomographic (i.e., diagonal) symbols. Inspired by the
tomographic procedures, in this section we will recover known tomographic
schemes, like the spin, the photon number and the symplectic tomographies,
from the groupoid immersions discussed in the previous section.

We will begin with spin tomography. Let us consider an irreducible unitary $%
\left( 2j+1\right) -$ dimensional representation $D^{j}\left( g\right) $ of
the group $SU\left( 2\right) $ and introduce the states $D^{j}\left(
g\right) \left\vert jm\right\rangle =:\left\vert g,jm\right\rangle .$ Here $%
g $\ is a suitable parametrization \cite{Ibort II} of the elements of $%
SU\left( 2\right) $. Then, using eq. \eqref{spin quant} we can write%
\begin{equation}
D^{j\dagger }\left( g\right) \hat{A}\,D^{j}\left( g\right)
=\sum_{m,m^{\prime }=-j}^{j}\tau _{D^{j}\left( g\right) \hat{A}D^{j\dagger
}\left( g\right) }^{j}\left( m,m^{\prime }\right) \left\vert m\right\rangle
\left\langle m^{\prime }\right\vert ,
\end{equation}%
which gives%
\begin{equation}
\hat{A}=\sum_{m,m^{\prime }=-j}^{j}w_{A}^{j}\left( g;m,m^{\prime }\right)
D^{j}\left( g\right) \left\vert jm\right\rangle \left\langle jm^{\prime
}\right\vert D^{j\dagger }\left( g\right)  \label{spin}
\end{equation}%
where
\begin{equation}
w_{A}^{j}\left( g;m,m^{\prime }\right) =\langle g,jm|\hat{A}|g,jm^{\prime
}\rangle =\mathrm{Tr}\left[ \hat{A}\left( D^{j}\left( g\right) \left\vert
jm\right\rangle \left\langle jm^{\prime }\right\vert D^{j\dagger }\left(
g\right) \right) ^{\dagger }\right] =\tau _{D^{j}\left( g\right) \hat{A}%
D^{j\dagger }\left( g\right) }^{j}  \notag
\end{equation}%
In other words, we get a new self--dual pair of quantizer-dequantizer, $%
\left\vert g,jm\right\rangle \left\langle g,jm^{\prime }\right\vert
=D^{j}\left( g\right) \left\vert jm\right\rangle \left\langle jm^{\prime
}\right\vert D^{j\dagger }\left( g\right) ,$\ which provides a new symbol $%
w_{A}^{j}\left( g;m,m^{\prime }\right) $ and a reconstruction formula, eq. %
\eqref{spin}, for $\hat{A}$. The diagonal part of the new symbol $%
w_{A}^{j}\left( g;m,m\right) $ is just the spin tomogram of $A,$ the $m-$th
\ component of a $\left( 2j+1\right) -$dimensional vector. When $A$ is a
density state operator, the vector is stochastic. Spin tomograms have their
own quantizer--dequantizer self--dual pair, extensively discussed in \cite%
{3m2000,Ibort II}.

As for the photon number tomography, let us consider the usual displacement
operator $\mathcal{D}\left( z\right) =\exp \left( z\hat{a}^{\dagger
}-z^{\ast }\hat{a}\right) ,$ with $\hat{a}^{\dagger },\hat{a}$ the creation
and annihilation operators, and $z$ a complex number. Then
\begin{equation}
\mathcal{D}^{\dagger }\left( z\right) \hat{A}\,\mathcal{D}\left( z\right)
=\sum_{n,m}F_{\mathcal{D}^{\dagger }\left( z\right) \hat{A}\mathcal{D}\left(
z\right) }\left( n,m\right) \left\vert n\right\rangle \left\langle
m\right\vert ,
\end{equation}%
which may be read, introducing $\left\vert n,z\right\rangle =\mathcal{D}%
\left( z\right) \left\vert n\right\rangle ,$ as:
\begin{eqnarray}
\hat{A} &=&\sum_{n,m}\Phi _{\hat{A}}\left( n,m;z\right) \mathcal{D}\left(
z\right) \left\vert n\right\rangle \left\langle m\right\vert \mathcal{D}%
^{\dagger }\left( z\right) \\
&=&\sum_{n,m}\Phi _{\hat{A}}\left( n,m;z\right) \left\vert n,z\right\rangle
\left\langle m,z\right\vert  \notag
\end{eqnarray}%
with%
\begin{eqnarray}
F_{\mathcal{D}^{\dagger }\left( z\right) \hat{A}\mathcal{D}\left( z\right)
}\left( n,m\right) &=&\mathrm{Tr}\left[ \mathcal{D}^{\dagger }\left(
z\right) \hat{A}\mathcal{D}\left( z\right) \left( \left\vert n\right\rangle
\left\langle m\right\vert \right) ^{\dagger }\right] \\
&=&\mathrm{Tr}\left[ \hat{A}\left( \left\vert n,z\right\rangle \left\langle
m,z\right\vert \right) ^{\dagger }\right] =\Phi _{\hat{A}}\left(
n,m;z\right) .  \notag
\end{eqnarray}%
Again, $\Phi _{\hat{A}}\left( n,m;z\right) $ is a symbol corresponding to a
new quantizer--dequantizer self--dual pair, given by
\begin{equation}
\mathcal{D}\left( z\right) \left\vert n\right\rangle \left\langle
m\right\vert \mathcal{D}^{\dagger }\left( z\right) =\left\vert
n,z\right\rangle \left\langle m,z\right\vert .
\end{equation}%
The diagonal part of the new symbol is nothing but the photon number
tomographic symbol $\mathcal{P}_{\hat{A}}\left( n,-z\right) $ (see for
instance \cite{MaMaToEur97}):%
\begin{equation}
\Phi _{\hat{A}}\left( n,n;z\right) =\mathcal{P}_{\hat{A}}\left( n,-z\right)
\end{equation}%
which admits the tomographic reconstruction formula%
\begin{equation}
\hat{A}=\int \frac{\mathrm{d}^{2}z}{\pi }\sum_{n=0}^{\infty }\mathcal{P}_{%
\hat{A}}\left( n,-z\right) \hat{D}_{\varphi }\left( n,z\right) ,
\end{equation}%
based on the following tomographic photon number quantizer:%
\begin{equation}
\hat{D}_{\varphi }\left( n,z\right) =\frac{2}{1-s}\left( \frac{s+1}{s-1}%
\right) ^{n}\hat{T}\left( z,-s\right) ,
\end{equation}%
where the $s-$ordered\ displaced parity operator, $\left( -1\leq s\leq
1\right) ,$ see \cite{Cahill-Glauber1969}%
\begin{equation}
\hat{T}\left( z,s\right) =\frac{2}{1-s}\mathcal{\hat{D}}\left( z\right)
\left( \frac{s+1}{s-1}\right) ^{\hat{a}^{\dagger }\hat{a}}\mathcal{\hat{D}}%
^{\dagger }\left( z\right) .
\end{equation}

As for symplectic tomography, let us consider the unitary operator
associated with a transformation $\hat{S}_{\mu \nu }$ (see \cite{Ibort
I,MaMaToPLA96} for more details):
\begin{eqnarray}
\hat{S}_{\mu \nu }\,\hat{q}\,\hat{S}_{\mu \nu }^{\dagger } &=&\mu \hat{q}%
+\nu \hat{p}, \\
\hat{S}_{\mu \nu }\,\hat{q}\,\hat{S}_{\mu \nu }^{\dagger }\hat{S}_{\mu \nu
}\left\vert x\right\rangle &=&x\hat{S}_{\mu \nu }\left\vert x\right\rangle
=x\left\vert x,\mu ,\nu \right\rangle ,  \notag
\end{eqnarray}%
with $\mu ,\nu $ real parameters. Then, by eq. \eqref{cont rec}:%
\begin{equation*}
\hat{S}_{\mu \nu }^{\dagger }\hat{A}\hat{S}_{\mu \nu }=\int \mathrm{d}x%
\mathrm{d}y\,\,f_{\hat{S}_{\mu \nu }^{\dagger }\hat{A}\hat{S}_{\mu \nu
}}\left( x,y\right) \left\vert x\right\rangle \left\langle y\right\vert
\end{equation*}%
or, equivalently:
\begin{equation}
\hat{A}=\int \mathrm{d}x\mathrm{d}y\,\,f_{\hat{S}_{\mu \nu }^{\dagger }\hat{A%
}\hat{S}_{\mu \nu }}\left( x,y\right) \hat{S}_{\mu \nu }\left\vert
x\right\rangle \left\langle y\right\vert \hat{S}_{\mu \nu }^{\dagger }
\end{equation}%
that we may write as%
\begin{equation}
\hat{A}=\int \mathrm{d}x\mathrm{d}y\,\,W_{\hat{A}}\left( x,y;\mu ,\nu
\right) \left\vert x,\mu ,\nu \right\rangle \left\langle y,\mu ,\nu
\right\vert ,
\end{equation}%
where%
\begin{eqnarray}
f_{\hat{S}_{\mu \nu }^{\dagger }\hat{A}\hat{S}_{\mu \nu }}\left( x,y\right)
&=&\mathrm{Tr}\left[ \hat{S}_{\mu \nu }^{\dagger }\hat{A}\hat{S}_{\mu \nu
}\left( \left\vert x\right\rangle \left\langle y\right\vert \right)
^{\dagger }\right] \\
&=&\mathrm{Tr}\left[ \hat{A}\left( \left\vert x,\mu ,\nu \right\rangle
\left\langle y,\mu ,\nu \right\vert \right) ^{\dagger }\right] =:W_{\hat{A}%
}\left( x,y;\mu ,\nu \right) .  \notag
\end{eqnarray}%
The new symbol $W_{\hat{A}}\left( x,y;\mu ,\nu \right) $ corresponds to a
new quantizer--dequantizer self--dual pair, i.e.,
\begin{equation}
\hat{S}_{\mu \nu }\left\vert x\right\rangle \left\langle y\right\vert \hat{S}%
_{\mu \nu }^{\dagger }=\left\vert x,\mu ,\nu \right\rangle \left\langle
y,\mu ,\nu \right\vert .
\end{equation}

By construction, the diagonal part of the new symbol is just the symplectic
tomogram $\mathcal{W}_{\hat{A}}\left( x,\mu ,\nu \right) $ of $\hat{A}$%
\begin{eqnarray}
W_{\hat{A}}\left( x,x;\mu ,\nu \right) &=&\mathrm{Tr}\left[ \hat{A}%
\left\vert x,\mu ,\nu \right\rangle \left\langle x,\mu ,\nu \right\vert %
\right] \\
&=&\mathrm{Tr}\left[ \hat{A}\delta (x\mathbb{I}-\mu \hat{q}-\nu \hat{p})%
\right] =:\mathcal{W}_{\hat{A}}\left( x,\mu ,\nu \right)  \notag
\end{eqnarray}%
On the other hand, associated with the dequantizer $\hat{U}_{\Sigma }\left(
x,\mu ,\nu \right) =\delta (x\mathbb{I}-\mu \hat{q}-\nu \hat{p}),$
symplectic tomography has its own quantizer $\hat{D}_{\Sigma }\left( x,\mu
,\nu \right) =\exp \left[ \mathrm{i}(x\mathbb{I}-\mu \hat{q}-\nu \hat{p})%
\right] /2\pi $, yielding the reconstruction formula

\begin{equation}
\hat{A}=\frac{1}{2\pi }\int \mathcal{W}_{A}(x,\mu ,\nu )\, e^{\mathrm{i}(x%
\mathbb{I}-\mu \hat{q}-\nu \hat{p})}\, \mathrm{d}x\mathrm{d}\mu \mathrm{d}%
\nu .  \label{symp rec}
\end{equation}

All in all, we have generalized the tomographic schemes introducing
tomographic quasi--distributions, whose diagonal part coincides with the
usual tomographies. Such tomographic quasi--distributions are obtained by
acting with families of unitary operators on certain self--dual
quantizer--dequantizer pairs, so they yield the same star--product kernels
of the generating ones, which in turn coincide with the convolution of some
groupoid algebras. However, the diagonal parts of these tomographic
quasi--distributions are in correspondence with their own tomographic
quantizer--dequantizer pairs which yield tomographic star--products, and in
general not coincident with the groupoid algebra convolutions. The two kind
of kernels can be related by means of the so called inter-twining maps and
dual symbols \cite{OlgaPatrizia}. Here, for instance, we limit ourselves to
show how an inter-twining map for the symplectic case, defined as the symbol
of $\hat{D}_{\Sigma }\left( x^{\prime },\mu ^{\prime },\nu ^{\prime }\right)
:$
\begin{equation}
W_{\hat{D}_{\Sigma }\left( x^{\prime },\mu ^{\prime },\nu ^{\prime }\right)
}\left( x,y;\mu ,\nu \right) =\mathrm{Tr}\left[ \hat{D}_{\Sigma }\left(
x^{\prime },\mu ^{\prime },\nu ^{\prime }\right) \left( \left\vert x,\mu
,\nu \right\rangle \left\langle y,\mu ,\nu \right\vert \right) ^{\dagger }%
\right] ,
\end{equation}%
allows to reconstruct the off--diagonal quasi-distribution by the diagonal
tomographic symbol. In fact, we may write, using eq. \eqref{symp rec}:%
\begin{eqnarray}
W_{\hat{A}}\left( x,y;\mu ,\nu \right) &=&\mathrm{Tr}\left[ \hat{A}\left(
\left\vert x,\mu ,\nu \right\rangle \left\langle y,\mu ,\nu \right\vert
\right) ^{\dagger }\right] \\
&=&\int \mathcal{W}_{A}(x^{\prime },\mu ^{\prime },\nu ^{\prime })W_{\hat{D}%
_{\Sigma }\left( x^{\prime },\mu ^{\prime },\nu ^{\prime }\right) }\left(
x,y;\mu ,\nu \right) \mathrm{d}x^{\prime }\mathrm{d}\mu ^{\prime }\mathrm{d}%
\nu ^{\prime }  \notag
\end{eqnarray}

\section{Conclusions}

To summarize we point out the main results of our paper. We have contributed
to clarify the role of groupoid structures in quantum mechanics by stating a
proposition connecting the convolution product for functions on a countable,
principal and transitive groupoid with the star--product scheme
corresponding to quantizer--dequantizer operators associated with Weyl
units. Such a proposition has been checked for finite and infinite Hilbert
spaces of qudits and modes of harmonic oscillators respectively. Also, we
have discussed a simple continuous case. Having done this, we have
considered a Fock space realization in a two mode system, to pave the way
for an extension to quantum field theory in a future work.

Known tomographic probability pictures were recognized as diagonal part of
quasi--distributions generated by means of groupoids in the tomographic
approach. We have also considered a groupoid associated with symplectic
tomography, so exhibiting the connection with phase--space and Weyl systems.

An Appendix, containing further mathematical details and generalizations of
the statements relating convolution-- and star--products has been provided.

We will end up this summary by recalling that category theory provides a
convenient setting to discuss groupoids, actually a groupoid is just a
(small) category all whose morphisms are invertible. It actually happens
that 2--groupoids, a deep generalization of groupoids from the point of view
of higher categories arises naturally when discussing the fundamental
structure of quantum mechanics and quantum information theory, as a matter
of fact they capture the mathematical structure of Schwinger's measurements
algebra \cite{Schwinger}. These issues will be considered in a forthcoming
paper.

\section*{Appendix}

\subsection*{A. Groupoids}

Bearing in mind the example of the pair--groupoid $\Gamma$ of sec. 2, we
will discuss here the definition and some properties of an abstract groupoid
$G$. To define a groupoid $G$ we need \cite{Connes,Renault}:

\begin{enumerate}
\item A set $G$ with a subset $G_{0}\subseteq G.$

\item Two maps $r$ and $s$ from $G$ onto $G_{0}$.

\item A binary operation $\circ $ (called multiplication) which is defined
for pairs $\gamma _{1}$ and $\gamma _{2}$ of elements in $G$ whenever $%
s(\gamma _{1}) = r(\gamma _{2})$ (we will say then that $\gamma_1$ and $%
\gamma_2$ are composable).
\end{enumerate}

Moreover $\left\{ r,s,\circ \right\} $ must satisfy the following axioms:

\begin{description}
\item[a)] $r(\gamma _{1}\circ \gamma _{2}) = r(\gamma _{1})$, $s(\gamma
_{1}\circ \gamma _{2})=s(\gamma _{2})$ for all composable $\gamma _{1}$ and $%
\gamma _{2}$.

\item[b)] $r(\gamma _{0})=s(\gamma _{0})=\gamma _{0}$, for all $\gamma
_{0}\in G_{0}.$

\item[c)] $r(\gamma)$, $s(\gamma)$ are left and right unities for $\gamma$
respectively, i.e., $r(\gamma )\circ \gamma =\gamma$, $\gamma \circ s(\gamma
)=\gamma$, for all $\gamma \in G.$

\item[d)] The multiplication $\circ $ is associative: if $(\gamma _{1}\circ
\gamma _{2})\circ \gamma _{3}$ is defined, then $\gamma _{1}\circ \left(
\gamma _{2}\circ \gamma _{3}\right) $ exists and $(\gamma _{1}\circ \gamma
_{2})\circ \gamma _{3}=\gamma _{1}\circ \left( \gamma _{2}\circ \gamma
_{3}\right) .$

\item[e)] Any $\gamma $ has a two--sided inverse $\gamma ^{-1},$ with $%
\gamma \circ \gamma ^{-1}=r(\gamma )$ and $\gamma ^{-1}\circ \gamma
=s(\gamma )$ . Moreover, $\left( \gamma ^{\prime }\circ \gamma \right) \circ
\gamma ^{-1}=\gamma ^{\prime }$ and $\gamma ^{-1}\circ (\gamma \circ \gamma
^{\prime \prime })=\gamma ^{^{\prime \prime }}$ for any $\gamma ^{\prime
},\gamma ^{\prime \prime }$ composable with $\gamma $. The map $\mathrm{inv}%
\colon \gamma \rightarrow \gamma ^{-1}$ is an involution, that is $\left(
\gamma ^{-1}\right) ^{-1}=\gamma $.
\end{description}

After condition e) $G_{0}$ is called the set of unities of $G$.

Some simple consequences can be easily derived from the above definition:

\begin{description}
\item[i)] It is possible to define in the set of unities $G_{0}$ the
following equivalence relation: $\gamma _{0}\sim \gamma _{0}^{\prime }$ \
iff there exists $\gamma \in G$ such that $r(\gamma )=\gamma _{0}$ and $%
s(\gamma )=\gamma _{0}^{\prime }$. Any equivalence class is called an orbit
of $G,$ and $G_{0}$ is the union of all these orbits.

\item[ii)] The set $H_{\gamma _{0}}=\{\gamma \in G:s(\gamma )=\gamma
_{0}=r(\gamma )\}$ is a group, called the isotropy group of $\gamma _{0}\in
G_{0}.$

\item[iii)] The isotropy groups of the unities of the same orbit in $G_{0}$
are isomorphic.
\end{description}

Two extreme cases regarding $G_{0}$ are:

\begin{description}
\item[$\protect\alpha )$] $G_{0}$ is the whole groupoid $G$: so that $%
s(\gamma )=r(\gamma )=\gamma $ $\forall \gamma \in G$ and any $\gamma $ can
be composed only with $\gamma $ yielding $\gamma \circ \gamma =\gamma $. The
orbits of $G$ are the single elements $\gamma $, the isotropy group of any $%
\gamma $ is trivial, containing only $\gamma $.

\item[$\protect\beta )$] The set of unities contains only one element, $%
G_{0}=\left\{ e\right\}$, so that $G$ is a group, there is only one orbit $%
\left\{ e\right\} $, and the isotropy group of $e$ is $G$.
\end{description}

Notice that it is also possible to define a groupoid $G$ starting with a set
$G_{0}$ not included in $G,$ so that $r,s$ map $G$ onto $G_{0}$ (which is
now called the base space of the groupoid). For any $x\in G_{0}$ consider
the subset $H_{x}=s^{-1}(x)\cap r^{-1}(x)\subset G.$ As $H_{x}$ is a group
containing the identity $e_{x}=\gamma ^{-1}\circ \gamma ,\forall \gamma \in
H_{x},$ we define the injection $i(x)=e_{x}$ of $G_{0}$ into $G,$ so that $%
H_{x}$ is just the isotropy group of $e_{x}$. Both ways of considering $%
G_{0} $ either as a subset of unities (inner) or as a base space (outer),
are equivalent. Each way can be more conveniently used according to the
circumstances.

It is also noticeable that there is an equivalent categorical definition of
a groupoid. Such approach will become particularly interesting in further
research on the structure of quantum systems that will be pursued elsewhere.

We illustrate all these points by considering the groupoids arising from the
action of a group on a set $G_{0}.$ To be concrete consider $SO(3)$ acting
on the two dimensional sphere $S^{2}$ and define $G=S^{2}\times SO(3)$ with $%
r(x,g)=(x,e)$ and $s(x,g)=(x^{g},e),$\ where $x^{g}=xg.$ Therefore $%
(x,g)\circ (y,g^{\prime })$ exists iff $y=x^{g}$ and $(x,g)\circ
(x^{g},g^{\prime }):=(x,gg^{\prime }).$ $G_{0}$ can be considered either as
a base space $S^{2}$ or as a subset $(S^{2},e)\subset G.$ As $S^{2}$ is a
homogeneous space under the action of $SO(3)$, there is only one orbit, the
sphere $S^{2}$ (or $(S^{2},e)\subset G$ ). The isotropy group of a point $%
x\in S^{2}$ is the subgroup of $SO(3)$ of rotations around the axis through $%
x,$ which is isomorphic to the isotropy groups of all others points of the
sphere.

Finally, a groupoid $G$ is called principal if the map
\begin{equation}
\left( r,s\right) :G\rightarrow G_{0}\times G_{0},\gamma \mapsto \left(
r\left( \gamma \right) ,s\left( \gamma \right) \right)
\end{equation}%
is one-to-one, it is called transitive if the map $\left( r,s\right) $ is
onto. For instance, the Cartesian product groupoid $\Gamma $ discussed in
sec. 2 is both principal and transitive.

\medskip

\begin{proposition}
\label{propunities} Two units $\gamma _{0},\gamma _{0}^{\prime }$ of the
same orbit are connected by a set of elements constituting a ``left coset" $%
\gamma H_{\gamma _{0}}$ of the isotropy group of $\gamma _{0}$ in $G:$%
\begin{eqnarray}
\gamma \circ \gamma _{0}\circ \gamma ^{-1} &=&\gamma _{0}^{\prime },\tilde{%
\gamma}\circ \gamma _{0}\circ \tilde{\gamma}^{-1}=\gamma _{0}^{\prime } \\
&\Leftrightarrow &\tilde{\gamma}=\gamma \circ \tilde{g},\ \tilde{g}\in
H_{\gamma _{0}}.  \notag
\end{eqnarray}
\end{proposition}

\begin{proof} If the groupoid is principal, the proposition is trivially
true, as $\tilde{\gamma}=\gamma .$ In general, chosen a $\gamma $ such that $%
\left( r\left( \gamma \right) ,s\left( \gamma \right) \right) =\left( \gamma
_{0}^{\prime },\gamma _{0}\right) ,$ observe that any $\tilde{\gamma}$ with $%
\left( r\left( \tilde{\gamma}\right) ,s\left( \tilde{\gamma}\right) \right)
=\left( \gamma _{0}^{\prime },\gamma _{0}\right) $ can be written as $%
g^{\prime }\circ \gamma \circ g,$where $g\in H_{\gamma _{0}}$ and $g^{\prime
}\in H_{\gamma _{0}^{\prime }}$ are elements of the isotropy groups of the
unities. These groups are isomorphic, so define $\gamma ^{-1}\circ g^{\prime
}\circ \gamma =:g_{1}\in H_{\gamma _{0}}.$\ Then $g^{\prime }=\gamma \circ
g_{1}\circ \gamma ^{-1}$ and we get:
\begin{equation}
\tilde{\gamma}=g^{\prime }\circ \gamma \circ g=\gamma \circ \tilde{g},\
\tilde{g}:=g_{1}\circ g\in H_{\gamma _{0}}.
\end{equation}%
Finally, we observe that an equivalent proposition holds when using the
``right coset" $H_{\gamma _{0}^{\prime }}\gamma $ of the isotropy group of $%
\gamma _{0}^{\prime }$, in place of the ``left coset" $\gamma H_{\gamma
_{0}}$.
\end{proof}

We remark that, as a set, any groupoid is the disjoint union of groupoids $%
G=\cup _{i}G_{i}$ corresponding to the partition of $G_{0}=\cup _{i}\mathcal{%
O}_{i}$ into orbits $\mathcal{O}_{i}$. A groupoid is transitive iff it has a
single orbit. Each $G_{i}$ has only one orbit of unities and elements in $%
G_{i}$ cannot be multiplied by elements in $G_{k},k\neq i.$ Furthermore, for
each $G_{i}$ the mapping%
\begin{equation}
\Psi :G_{i}\rightarrow \mathcal{O}_{i}\times \mathcal{O}_{i}:\gamma \mapsto
(r(\gamma ),s(\gamma ))  \label{morphism}
\end{equation}%
is a morphism of $G_{i}$ onto the principal groupoid. All elements in $G$
belonging to\ $\cup _{x\in \mathcal{O}_{i}}\left\{ \gamma \right\} _{x}$,
where $\left\{ \gamma \right\} _{x}$ is the isotropy group of $x\in \mathcal{%
O}_{i},$ are mapped onto the diagonal of $\mathcal{O}_{i}\times \mathcal{O}%
_{i}:\left\{ \gamma \right\} _{x}\mapsto (x,x),$ they constitutes the kernel
of this morphism. Elements of different $G_{i}$'s may be related by
introducing topological requirements on $G$.

\subsection*{B. Finite groupoids}

We restrict now to the case of a finite transitive groupoid of order $K:$
\begin{equation}
G=\left\{ \gamma _{k}\right\} _{k=1}^{K},\mathrm{ord}\left( G\right) =K.
\end{equation}%
The groupoid algebra $\mathcal{F}\left( G\right) $ is the algebra of the
(complex or real) functions on the groupoid with the convolution product
\begin{equation}
\left( f_{1}\ast f_{2}\right) (\gamma _{i})=\sum_{\overset{j,k}{{\gamma
_{j}\circ \gamma }_{k}{=\gamma }_{i}}}f_{1}(\gamma _{j})f_{2}(\gamma _{k}).
\end{equation}%
Hereafter any summation label ranges from $1$ to $K$.The functions $\delta
_{\gamma _{j}},$ defined as%
\begin{equation}
\delta _{\gamma _{j}}(\gamma _{k})=\left\{
\begin{array}{c}
1\ \mathrm{if}\ \gamma _{j}=\gamma _{k} \\
0\ \ \mathrm{if}\ \gamma _{j}\neq \gamma _{k}%
\end{array}%
\right. ,
\end{equation}%
are a basis of the groupoid algebra. They can be associated with the
standard basis of $K-$dimensional column vectors: $\delta _{\gamma
_{j}}\mapsto v_{j},$with $\left( v_{j}\right) _{k}=\delta _{jk}.$ For any $%
f\in \mathcal{F}\left( G\right) :$

\begin{equation}
f\left( \cdot \right) =\sum\limits_{k}f\left( \gamma _{k}\right) \delta
_{\gamma _{k}}(\cdot ).
\end{equation}%
Moreover:
\begin{equation}
\left( \delta _{\gamma _{j}}\ast f\right) (\gamma _{i})=\sum\limits_{\overset%
{k}{\gamma _{j}\circ \gamma _{k}=\gamma _{i}}}f(\gamma _{k}).
\end{equation}%
In particular
\begin{equation}
\left( \delta _{\gamma _{j}}\ast \delta _{\gamma _{h}}\right) (\gamma
_{i})=\sum\limits_{\overset{k}{\gamma _{j}\circ \gamma _{k}=\gamma _{i}}%
}\delta _{\gamma _{h}}(\gamma _{k})=\sum\limits_{\gamma _{j}\circ \gamma
_{k}=\gamma _{i}}1.
\end{equation}%
In other words, $\delta _{\gamma _{j}}\ast \delta _{\gamma _{h}}$ is $1$ in $%
\gamma _{i}=\gamma _{j}\circ \gamma _{k}$ (if $\gamma _{j}$ and $\gamma _{k}$
are composable) and is $0$ elsewhere; so
\begin{equation}
\delta _{\gamma _{j}}\ast \delta _{\gamma _{h}}=\delta _{\gamma _{j}\circ
\gamma _{k}}.
\end{equation}%
The above equation shows that the convolution product is associative because
the multiplication $\circ $ is associative.

The mapping
\begin{equation}
\gamma \rightarrow D(\gamma )=\delta _{\gamma }\ast  \label{deltagr}
\end{equation}%
is a groupoid realization by means of operators in $L(\mathcal{F}\left(
G\right) )$ in the following sense:%
\begin{equation}
D(\gamma _{j})D(\gamma _{k})=\left\{
\begin{array}{c}
D(\gamma _{j}\circ \gamma _{k})\quad \mathrm{if\quad }\gamma _{j}\circ
\gamma _{k}\mathrm{\quad exists,} \\
0\quad \mathrm{if\quad }\gamma _{j}\circ \gamma _{k}\mathrm{\quad
does~not~exist.}%
\end{array}%
\right.
\end{equation}%
In general, we may define a groupoid realization as a morphism $\Phi $ of
the groupoid $G$ in the set of operators $L(V)$ on some linear space $V$,
such that the above equations are satisfied. This means, in abstract, that
we add to the groupoid a zero element which is the result of any forbidden
multiplication. The realization of the groupoid is then a realization in the
usual sense of this new enlarged structure. In sec. 2 we have introduced an
irreducible realization using the Weyl units $E^{\prime }$s.

Given the groupoid realization $D$, a realization of the groupoid algebra is
immediately obtained by the formula

\begin{equation}
A_{f}=\sum\limits_{k}f\left( \gamma _{k}\right) D(\gamma _{k}).
\label{D quant}
\end{equation}%
which looks like a quantization formula, where an operator $A_{f}$ is
obtained by a function $f$ by means of a quantizer $D(\gamma _{k}).$ From
the above equation we obtain:
\begin{eqnarray}
A_{f_{1}}A_{f_{2}} &=&\sum_{j,k}f_{1}(\gamma _{j})f_{2}(\gamma _{k})D(\gamma
_{j}\circ \gamma _{k})=\sum_{\overset{i,j,k}{{\gamma _{j}\circ \gamma }_{k}{%
=\gamma }_{i}}}f_{1}(\gamma _{j})f_{2}(\gamma _{k})D(\gamma _{i}) \\
&=&\sum_{i}\left( f_{1}\ast f_{2}\right) \left( \gamma _{i}\right) D(\gamma
_{i})=A_{f_{1}\ast f_{2}}.  \notag
\end{eqnarray}%
In other terms, the product of operators corresponds to the convolution
product of the associated functions. In the standard basis, which we will
use from now on, operators are $K\times K-$matrices acting on the $K-$%
dimensional vector space of the functions.

Moreover, we can state the following:

\begin{lemma}
For a transitive groupoid, the dequantizer associated to the quantizer $%
D(\gamma _{k})$\textit{\ is }$D(\gamma _{k})/\mathcal{N}$, where $\mathcal{N}
$ is a suitable normalization constant, with
\begin{eqnarray}
\frac{1}{\mathcal{N}}\mathrm{Tr}\left[ D(\gamma _{j})D^{T}(\gamma _{k})%
\right] &=&\delta _{\gamma _{j}}(\gamma _{k})  \label{Tr} \\
\mathcal{N} &=&\mathrm{ord}\left( G_{0}\right) +\mathrm{ord}\left( H_{\gamma
_{j}^{-1}\circ \gamma _{j}}-\gamma _{j}^{-1}\circ \gamma _{j}\right) .
\label{N}
\end{eqnarray}
\end{lemma}

\begin{proof} In the standard basis the trace can be written as%
\begin{eqnarray}
\sum_{p,q}\left( D(\gamma _{j})\right) _{pq}\left( D^{T}(\gamma _{k})\right)
_{qp} &=&\sum_{p,q}\delta _{\gamma _{j}\circ \gamma _{p}}(\gamma _{q})\delta
_{\gamma _{k}\circ \gamma _{p}}(\gamma _{q}) \\
&=&\delta _{\gamma _{j}}(\gamma _{k})\sum_{p,q}\delta _{\gamma _{j}\circ
\gamma _{p}}(\gamma _{q}).  \notag
\end{eqnarray}%
The evaluation of the last summation amounts essentially to count all the $%
\gamma _{p}$'s that can be composed with$\ \gamma _{j}:$ they are those
contained in the isotropy group of $\gamma _{j}^{-1}\circ \gamma _{j},$ plus
the elements connecting $\gamma _{j}^{-1}\circ \gamma _{j}$ with all the
other unities of the unique orbit. The evaluation, from Proposition \ref{propunities},
yields eventually the normalization constant $\mathcal{N}$ of eq. \eqref%
{N}.
\end{proof}

In view of eq. \eqref{D quant}, we may write%
\begin{equation}
f(\gamma _{k})=\frac{1}{\mathcal{N}}\mathrm{Tr}\left[ A_{f}D^{T}(\gamma _{k})%
\right] ,
\end{equation}%
and the symbol functions may be multiplied using a star--product kernel $%
K\left( \gamma _{j},\gamma _{k},\gamma _{i}\right) .$ We have:

\begin{proposition}
\label{gen_conv} The star--product corresponding to the kernel%
\begin{equation}
K\left( \gamma _{j},\gamma _{k},\gamma _{i}\right) =\frac{1}{\mathcal{N}}%
\mathrm{Tr}\left[ D(\gamma _{j})D(\gamma _{k})D^{T}(\gamma _{i})\right]
\end{equation}%
coincides with the convolution product of the groupoid algebra.
\end{proposition}

\begin{proof}  In view of eq. \eqref{Tr}, the evaluation of
the kernel gives%
\begin{equation}
K\left( \gamma _{j},\gamma _{k},\gamma _{i}\right) =\delta _{\gamma
_{j}\circ \gamma _{k}}(\gamma _{i}),
\end{equation}%
so that the star--product
\begin{eqnarray}
\left( f_{1}\star f_{2}\right) (\gamma _{i}) &=&\sum_{j,k}K\left( \gamma
_{j},\gamma _{k},\gamma _{i}\right) f_{1}(\gamma _{j})f_{2}(\gamma _{k}) \\
&=&\sum_{\overset{j,k}{{\gamma _{j}\circ \gamma }_{k}{=\gamma }_{i}}%
}f_{1}(\gamma _{j})f_{2}(\gamma _{k})=\left( f_{1}\ast f_{2}\right) (\gamma
_{i})  \notag
\end{eqnarray}%
is nothing but the convolution product.
\end{proof}

The above proposition generalizes Proposition 1 (sec. 4) to the case of a
transitive, nonprincipal groupoid. We observe that when $G$ is also
principal with $N$ equivalent unities, its order $K=N^{2}$ and $D-$%
realization is reducible and contains $N$ times the irreducible realization
of the Weyl units.

\subsection*{C. Remarks on the convolution in the continuous case}

Let us consider now a measure space $\left( S,\mu \right) $ and the
principal, transitive pair groupoid $\Gamma =S\times S,$ with the usual
composition law $(x,y)\circ \left( y,z\right) =\left( x,z\right) $. It is
natural to choose for $S\times S$\ the product measure $\mu \times \mu ,$
even though this is not mandatory. The convolution formula would read now:
\begin{equation}
\left( f_{1}\ast f_{2}\right) \left( x,y\right) =\int\limits_{\overset{%
y_{1}=x_{2}}{x_{1}=x,y_{2}=y}}f_{1}(x_{1},y_{1})f_{2}(x_{2},y_{2})\mu (%
\mathrm{d}x_{1})\mu (\mathrm{d}y_{1})\mu (\mathrm{d}x_{2})\mu (\mathrm{d}%
y_{2})  \label{convolution}
\end{equation}%
This is an integral on the subset $M$ of $S\times S$ $\times S\times S$
given by the constraints $y_{1}=x_{2},x_{1}=x,y_{2}=y,$ induced by the
groupoid composition law and this integral in general will be zero. As for a
countable groupoid, with $S=\{x_{k},k=1,\dots ,N\leq \infty \},\mu
=\sum_{k}\mu _{k},$ where $\mu _{k}( \Delta ) =1$ if $x_{k}\in \Delta$ and
zero otherwise, yielding the usual convolution formula $\left( \ref{conv}%
\right) $. Therefore it is necessary to make sense of the integral above. To
discuss a concrete relevant case, assume $S=\mathbb{R}$, and $\mu \left(
\mathrm{d}x\right) =\mathrm{d}x,$ the Lebesgue measure on $\mathbb{R}$. The
groupoid composition constraint becomes the linear constraints $%
h_{1}(x_{1},x)=x_{1}-x=0,\ $\ $h_{2}(y_{1},x_{2})=y_{1}-x_{2}=0,$ $%
h_{3}(y_{2},y)=y_{2}-y=0$ and then, the convolution integral may be read as%
\begin{eqnarray}
\left( f_{1}\ast f_{2}\right) \left( x,y\right) &=&\int
f_{1}(x_{1},y_{1})f_{2}(x_{2},y_{2})\delta (x_{1}-x)\delta
(y_{1}-x_{2})\delta (y_{2}-y)\mathrm{d}x_{1}\mathrm{d}y_{1}\mathrm{d}x_{2}%
\mathrm{d}y_{2}  \notag \\
&=&\int_{\mathbb{R}}f_{1}(x,y_{1})f_{2}(y_{1},y)\mathrm{d}y_{1}
\end{eqnarray}%
which is just the continuous matrix product of eq. \eqref{contmat} where $%
f_{1},f_{2}$ are continuous functions with proper support ($f$ is a
continuous function with proper support iff for each compact subset $K$ of $%
\mathbb{R}$, the intersections of $K\times \mathbb{R}$ and $\mathbb{R}\times
K$ with the set $\left\{ (x;y)|f(x;y)=0\right\} $ have compact closure).

However, this way to write the convolution is arbitrary to some extent. In
fact, the form above depends on the choice of linear constraints, leading to
$\delta (x-y)$ instead of $\delta (h(x-y))$ with an arbitrary function such
that $h(0)=0,h^{\prime }(0)\neq 0$, realizing the same constrains. For
instance, upon substituting $\delta (x_{1}-x)\rightarrow \delta (\alpha
_{1}(x)(x_{1}-x))$ , $\delta (y_{1}-x_{2})\rightarrow \delta (y_{1}-x_{2})$
, $\delta (y_{2}-y)\rightarrow \delta (\alpha _{2}(y)(y_{2}-y))$ \ with $%
\alpha _{1,2}>0$ the integration domain $M=\mathbb{R}$ does not change, but
the convolution formula becomes%
\begin{equation}
\left( f_{1}\ast f_{2}\right) \left( x,y\right) =\frac{1}{\alpha
_{1}(x)\alpha _{2}(y)}\int_{\mathbb{R}}f_{1}(x,y_{1})f_{2}(y_{1},y)\mathrm{d}%
y_{1}
\end{equation}

Moreover, a different measure can be chosen. Suppose  $\mu \left( \mathrm{d}%
x\right) =\rho \left( x\right) \mathrm{d}x,$ then the space of integrable
groupoid functions varies accordingly and the convolution formula reads%
\begin{equation}
\left( f_{1}\ast f_{2}\right) \left( x,y\right) =\rho \left( x\right) \rho
\left( y\right) \int_{\mathbb{R}}f_{1}(x,y_{1})f_{2}(y_{1},y)\rho ^{2}\left(
y_{1}\right) \mathrm{d}y_{1}.
\end{equation}

These considerations holds for any principal transitive groupoid with a base
space $G_0$ which is a measure space $\left( S,\mu \right) $; in this case
the morphism of eq. \eqref{morphism} is an isomorphism and one easily
concludes that for this kind of groupoids one can always reduce the
convolution to the form%
\begin{equation}
\left( f_{1}\ast f_{2}\right) \left( x,y\right) =\rho \left( x\right) \rho
\left( y\right) \int_{s \in S}f_{1}(x,s)f_{2}(s,y)\rho ^{2}\left( s\right)
\mu(s).
\end{equation}
To illustrate this point consider the groupoid arising from the action of
the group $\mathbb{R}$ as translations on the base space $\left( \mathbb{R},%
\mathrm{d}x\right) $ (analogous to the action of $SO(3)$ discussed
previously).

Using again the delta functions as constraints, the general expression of
convolution \eqref{convolution} yields%
\begin{equation}
\left( f_{1}\ast f_{2}\right) \left( x,t\right) =\int_{\mathbb{R}%
}f_{1}(x,s-x)f_{2}(s,t+x-s)\mathrm{d}s
\end{equation}%
The morphism $\left( \ref{morphism}\right) $\ for this example reads $\Psi
(x,t)=(x,t+x)$ and introducing the pullback $\Psi ^{\ast }f$ of a function $f
$ we get
\begin{equation}
\Psi ^{\ast }\left( f_{1}\ast f_{2}\right) \left( x,t\right) =\Psi ^{\ast
}\int_{\mathbb{R}}f_{1}(x,s-x)f_{2}(s,t+x-s)\mathrm{d}s=\int_{\mathbb{R}%
}\Psi ^{\ast }f_{1}(x,s)\Psi ^{\ast }f_{2}(s,t)\mathrm{d}s
\end{equation}%
and recover the usual form of convolution.

In the general case the groupoid elements will have an isotropy group $%
H_{\gamma _{0}}$ and the morphism $\Psi $ therefore will have a kernel. To
integrate functions on such a groupoid one would need a measure on the set
of units and a measure on the isotropy group. For the isotropy group it is
natural to use the Haar measure, when available, and one should choose some
measure on the set of units $G_{0}$.

We will just mention here two cases:

\begin{enumerate}
\item $G$ a transitive groupoid, $G_{0}=$ \{ set of $N$ elements \}, and $%
H_{\gamma _{0}}$ a group with an invariant measure. In this case it can be
shown that functions on $G$ are represented as $N\times N-$matrices, whose
entries are functions on $H_{\gamma _{0}}$. Convolution of two functions is
now a row--by--column product of the corresponding matrices, where products
of matrix elements are replaced by the usual group convolution on $H_{\gamma
_{0}}$.

\item $G$ is the groupoid arising from the action of the rigid motion on a
plane. In this case $G_{0}$ is the plane and $H_{\gamma _{0}}$ is the
rotation around an axis. It is natural to use the Haar measure $\mathrm{d}%
\theta $ for $H_{\gamma _{0}}$ and $\mathrm{d}x\, \mathrm{d}y$ on the plane
which is invariant under rigid motions. As before the convolution can be
written as a continuous matrix multiplication, but with a group convolution
in the $\theta $ variable.
\end{enumerate}

In the general case of a topological groupoid $G$ the convolution algebra
can be constructed by using a family of measures as it is done for instance
in \cite{Renault} where a left Haar system of measures is used, however
there is more elegant construction using half--densities as discussed for
instance in \cite{Connes} (Sec. II.5, page 101) but that goes back to the
work of S. Zakrzewski and P. Stachura \cite{Stachura}.

\bigskip

\paragraph{Acknowlegments}

This work was partially supported by MEC grants FPA--2009--09638,
MTM2010--21186--C02--02, QUITEMAD programme and DGA--E24/2.  
A.I. wants to acknowledge the support provided by the
Program  P2009 ESP-1594, and ``Programa
Salvador de Madariaga''.  
G. M. would like to acknowledge the support provided by the Santander/UCIIIM Chair of Excellence programme 2011-2012.

\end{document}